\documentclass[10pt]{amsart}
\usepackage{amssymb}
\usepackage{amsmath,amssymb,amsfonts,amsthm,
latexsym, amscd, amsfonts, epsfig}
\usepackage{mathrsfs}
\usepackage{color}
\usepackage{eucal}%    caligraphic-euler fonts: \mathcal{ }
\usepackage{eufrak}%   frak-euler        fonts: \mathfrak{ }
\usepackage[all]{xypic}
\usepackage{xspace}

\theoremstyle{plain}
\newtheorem{theorem}{Theorem}

\newtheorem{lemma}{Lemma}

\theoremstyle{definition}
\newtheorem{definition}{Definition}

\theoremstyle{example}

\theoremstyle{remark}

\numberwithin{equation}{section}

\begin{document}

%%%
%%%
%%%%%%%%%%%%%%%%%%%%%%%%%%%%%%%%%%%%%%%%%%%%%%%%%%%%%%%%%%%%%%%%%%%%%%%%%%
%%
%%%
\title[Asymptotic Enumeration of RNA Structures with Pseudoknots]
      {Asymptotic Enumeration of RNA Structures with Pseudoknots}
\author{Emma Y. Jin and Christian M. Reidys$^{\,\star}$}
\address{Center for Combinatorics, LPMC-TJKLC \\
         Nankai University  \\
         Tianjin 300071\\
         P.R.~China\\
         Phone: *86-22-2350-6800\\
         Fax:   *86-22-2350-9272}
\email{reidys@nankai.edu.cn}
\thanks{}
\date{June 2007}
\begin{abstract}
In this paper we present the asymptotic enumeration of RNA structures
with pseudoknots. We develop a general framework for the computation
of exponential growth rate and the sub exponential factors for
$k$-noncrossing RNA structures.
Our results are based on the generating function for the number of
$k$-noncrossing RNA pseudoknot structures, ${\sf S}_k(n)$, derived in
\cite{Reidys:07pseu}, where $k-1$ denotes the maximal size of sets of
mutually intersecting bonds.
We prove a functional equation for the generating function $\sum_{n\ge
0}{\sf S}_k(n)z^n$ and obtain for $k=2$ and $k=3$ the analytic
continuation and singular expansions, respectively.
It is implicit in our results that for arbitrary $k$
singular expansions exist and via transfer theorems of analytic
combinatorics we obtain asymptotic expression for the coefficients.
We explicitly derive the
asymptotic expressions for $2$- and $3$-noncrossing RNA structures.
Our main result is the derivation of the formula ${\sf S}_3(n)  \sim
\frac{10.4724\cdot 4!}{n(n-1)\dots(n-4)}
                    \left(\frac{5+\sqrt{21}}{2}\right)^n$.
\end{abstract}
\maketitle
{{\small
%\tableofcontents
}}

%%%
%%%
%%%%%%%%%%%%%%%%%%%%%%%%%%%%%%%%%%%%%%%%%%%%%%%%%%%%%%%%%%%%%%%%%%%%%%%%
%%%
%%%

\section{Introduction}

%%%
%%%%%%%%%%%%%%%%%%%%%%%%%%%%%%%%%%%%%%%%%%%%%%%%%%%%%%%%%%%%%%%%%%%%%%%%
%%%

RNA molecules are particularly fascinating since they represent
both: genotypic legislative via their primary sequence and
phenotypic executive via their functionality associated to 2D or
3D-structures, respectively. Accordingly, it is believed that RNA
may have been instrumental for early evolution--before Proteins
emerged. The primary sequence of an RNA molecule is the sequence of
nucleotides {\bf A}, {\bf G}, {\bf U} and {\bf C} together with the
Watson-Crick ({\bf A-U}, {\bf G-C}) and ({\bf U-G}) base pairing
rules specifying the pairs of nucleotides can potentially form
bonds. Single stranded RNA molecules form helical structures whose
bonds satisfy the above base pairing rules and which, in many cases,
determine their function. For instance, RNA ribozymes are capable of
catalytic activity, cleaving other RNA molecules. RNA secondary
structure prediction is of polynomial complexity \cite{Waterman:78a}
which is result from the fact that in secondary structures no two
bonds can cross. Leaving the paradigm of RNA secondary structures,
i.e.~studying RNA structures with crossing bonds, the RNA pseudoknot
structures, poses challenging problems for computational biology.
Prediction algorithms for RNA pseudoknot structures are much harder
to derive since there exists no {\it a priori} tree-structure and
the subadditivity of local solutions is not guaranteed. RNA
pseudoknot structures can be categorized in terms of the maximal
size of sets of mutually crossing bonds \cite{Reidys:07pseu}. To be
precise a $k$-noncrossing RNA structure has at most $k-1$ mutually
crossing bonds and a minimum bond-length of $2$, i.e.~for any $i$,
the nucleotides $i$ and $i+1$ cannot form a bond. The asymptotics of
$k$-noncrossing RNA structures is of central importance in this
context. The key question is how to decompose a $k$-noncrossing RNA
structure into a collection of sub-structures (which can easily be
computed), and what are the properties of this decomposition. Given
such a decomposition we can predict the factors and reassemble the
corresponding pseudoknot structure. A first step towards finding
such decompositions is to have information about the cardinalities
of the respective sets of structures involved.
%%%
%%%%%%%%%%%%%%%%%%%%%%%% Figures ex1 %%%%%%%%%%%%%%%%%%%%%%%%%%%%%%%%%%%%%%
%%%
\begin{figure}[ht]
\centerline{%
\epsfig{file=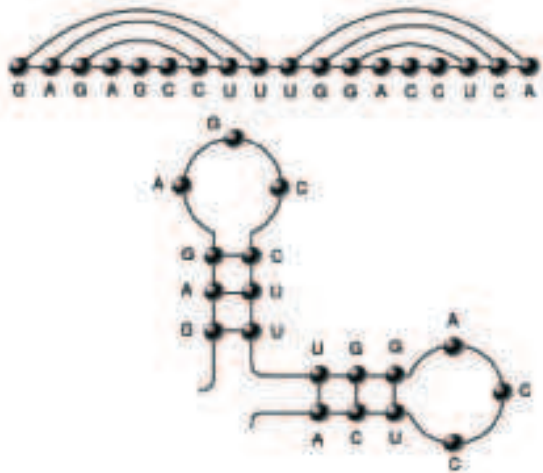,width=0.8\textwidth}\hskip15pt
 }
\caption{\small RNA secondary structures. Diagram representation
(top): the primary sequence, {\bf GAGAGCCUUUGGACCUCA}, is drawn
horizontally and its backbone bonds are ignored. All bonds are drawn
in the upper halfplane and secondary structures have the property
that no two arcs intersect and all arcs have minimum length $2$.
Outer planar graph representation (bottom). } \label{F:1}
\end{figure}
%%%
%%%%%%%%%%%%%%%%%%%%%%%%%%%%%%%%%%%%%%%%%%%%%%%%%%%%%%%%%%%%%%%%%%%%%%%%%
%%%
The asymptotic analysis of $k$-noncrossing RNA structures is based on
their generating function, obtained in \cite{Reidys:07pseu}. The particular
formulas are however alternating sums, which make even the computation of
the exponential growth rate a nontrivial task. In this paper we develop a
framework for the asymptotic enumeration of $k$-noncrossing RNA structures.
Before we go into this in more detail, let us first provide some background
on coarse grained RNA structures and put our results into context.

\subsection{RNA secondary structures or the universality of the square root}
About three decades ago Waterman {\it et.al.} pioneered the concept
of RNA secondary structures
\cite{Penner:93c,Waterman:79a,Waterman:78a,Waterman:94a,Waterman:80}.
The key property of secondary structures is best understood,
considering a structure as a diagram, which is obtained as follows:
one draws the primary sequence of nucleotides horizontally and
ignores all chemical bonds of its backbone. Then one draws all
bonds, i.e.~nucleotide interactions satisfying the Watson-Crick base
pairing rules (and {\bf G}-{\bf U} pairs) as arcs in the upper
halfplane, effectively identifying structure with the set of all
arcs. In this representation, RNA secondary structures have then
following property: there exist no two arcs $(i_1,j_1)$,
$(i_2,j_2)$, where $i_1<j_1$ and $i_2<j_2$ with the property
$i_1<i_2<j_1<j_2$ and all arcs have at least length $2$.
Equivalently, there exist no two arcs that cross in the diagram
representation of the structure, see Figure~\ref{F:1}.
%%%
%%%%%%%%%%%%%%%%%%%%%%%% Figures ex1 %%%%%%%%%%%%%%%%%%%%%%%%%%%%%%%%%%%%%%
%%%
\begin{figure}[ht]
\centerline{%
\epsfig{file=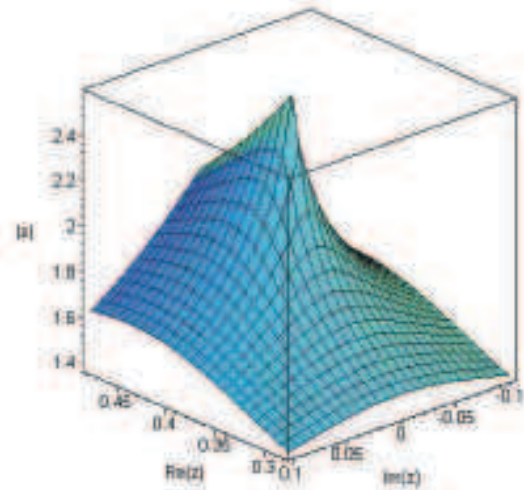,width=0.55\textwidth}\hskip15pt}
\caption{\small Universality of the square root. We display the
branch-point singularity (here at $\rho_2=\frac{3-\sqrt{5}}{2}$),
i.e.~the critical singularity for the asymptotics of RNA secondary structures.
{\it All} singularities arising from enumeration of certain classes of
secondary structures produces this type, whence the sub exponential factor
$n^{-\frac{3}{2}}$.
}
\label{F:1b}
\end{figure}
%%%
%%%%%%%%%%%%%%%%%%%%%%%%%%%%%%%%%%%%%%%%%%%%%%%%%%%%%%%%%%%%%%%%%%%%%%%%%
%%%
Basically, all combinatorial properties of secondary structures are derived
from Waterman's basic recursion \cite{Waterman:78a}
\begin{equation}\label{E:basic}
{\sf S}_2(n)= {\sf S}_2(n-1)+\sum_{s=0}^{n-2}
{\sf S}_2(n-2-s){\sf S}_2(s) \ ,
\end{equation}
where ${\sf S}_2(n)$ denotes the number of RNA secondary structures.
Eq.~(\ref{E:basic}) is an immediate consequence considering
secondary structures as Motzkin paths, i.e.~peak-free paths with
{\it up}, {\it down} and {\it horizontal} steps that stay in the
upper halfplane, starting at the origin and end on the $x$-axis. The
recursion is in particular the key for all asymptotic results since
it allows to obtain an implicit function equation for the generating
function and subsequent application of Darboux-type theorems
\cite{Schuster:98a,Wong:74}. If specific conditions are being
imposed, for instance minimum loop-size or stack length, it is
straightforward to translate these constraints into restricted
Motzkin paths, all of which satisfy some variant of
eq.~(\ref{E:basic}). As a result, all asymptotic formulae are of the
same type: a square root, that is, the asymptotic behavior is
determined by an algebraic branch singularity with the
sub exponential factor $n^{-\frac{3}{2}}$. For instance, the number
of RNA secondary structures having a minimum hairpin-loop length of
$3$ and minimum stack-length $2$ is asymptotically given by ${\sf
S}_2(n)\sim 1.4848 \, n^{-\frac{3}{2}} 1.8488^n$
\cite{Schuster:98a}. The number of RNA secondary structures having
exactly $\ell$ isolated vertices, ${\sf S}_2(n,\ell)$, satisfies the
two term recursion $(n-\ell)(n-\ell+2)\, {\sf S}_{2}(n,\ell)\,
-\, (n+\ell)(n+\ell-2)\,{\sf S}_{2}(n-2,\ell)=0$
\cite{Reidys:07pseu} and Waterman proved in \cite{Waterman:94a} the
following beautiful formula
\begin{equation}\label{E:Waterman-tree}
{\sf S}_2(n,\ell)  =
      \frac{2}{n-\ell}{\frac{n+\ell}{2} \choose \frac{n-\ell}{2} +1}
                    {\frac{n+\ell}{2}-1 \choose \frac{n-\ell}{2}-1}
\end{equation}
resulting from a bijection between secondary structure and linear trees.
In \cite{Waterman:86} it is shown that the prediction of secondary
structures can be obtained in polynomial time and yet again
eq~(\ref{E:basic}) is central for all folding algorithms
\cite{Zuker:79b,Schuster:98a,Waterman:86,Bauer:96,Tacker:94a,McCaskill:90a}.

\subsection{Beyond secondary structures}
While the concept of secondary structure is of fundamental
importance it is well-known that there exist additional types of
nucleotide interactions \cite{Science:05a}. These bonds are called
pseudoknots \cite{Westhof:92a} and occur in functional RNA (RNAseP
\cite{Loria:96a}), ribosomal RNA \cite{Konings:95a} and are
conserved in the catalytic core of group I introns.
%%%
%%%%%%%%%%%%%%%%%%%%%%%% Figures bisec %%%%%%%%%%%%%%%%%%%%%%%%%%%%%%%%%%%%%%
%%%
\begin{figure}[ht]
\centerline{%
\epsfig{file=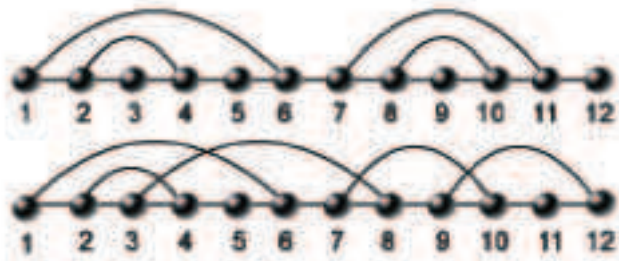,width=0.85\textwidth}
 }
\caption{\small Beyond secondary structures: an RNA bi-secondary structure
as the generalization from outer-planar to planar diagrams.
We display a secondary RNA structure (top) and a bi-secondary structure
(bottom). Reflecting the arcs $(3,8)$ and $(9,12)$ w.r.t. the $x$-axis
yields two secondary structures.}
\label{F:3}
\end{figure}
%%%
%%%%%%%%%%%%%%%%%%%%%%%%%%%%%%%%%%%%%%%%%%%%%%%%%%%%%%%%%%%%%%%%%%%%%%%%%
%%%
In plant viral RNAs pseudoknots mimic tRNA structure
and in {\it in vitro} RNA evolution \cite{Tuerk:92} experiments have produced
families of RNA structures with pseudoknot motifs, when binding HIV-1
reverse transcriptase. Important mechanisms like ribosomal frame shifting
\cite{Chamorro:91a} also involve pseudoknot interactions.
There exist several prediction algorithms for pseudoknot RNA structures
\cite{Rivas:99a,Uemura:99a,Akutsu:00a,Lyngso:00a} all of which can identify
particular respective pseudoknot motifs.
Stadler {\it et al.} \cite{Stadler:99a} suggested a classification of
RNA pseudoknot-types based on a notion of inconsistency graphs and
computed the upper bound of $4.7613$ for the exponential growth factor
of bi-secondary structures. Bi-secondary structures
are ``superpositions'' of two secondary structures, i.e.~they can be drawn
as a set of non intersecting arcs in the upper and lower half plane,
respectively.
Figure~\ref{F:3} shows how bi-secondary structures naturally
arise when passing from outer-planar to planar diagram representations.
The concept of $k$-noncrossing RNA structures generalize both: secondary
and bi-secondary structures, respectively.
While RNA secondary structures are precisely $2$-noncrossing RNA structures,
bi-secondary structures correspond to planar $3$-noncrossing RNA structures.
The key advantage of $k$-noncrossing RNA structures is that their defining
property is intrinsically local. It can be expected that this facilitates
fast folding algorithms. In Figure~\ref{F:4} we contrast all three structural
concepts, secondary, bi-secondary and $k$-noncrossing RNA structures.
%%%
%%%%%%%%%%%%%%%%%%%%%%%% Figures ex1 %%%%%%%%%%%%%%%%%%%%%%%%%%%%%%%%%%%%%%
%%%
\begin{figure}[ht]
\centerline{%
\epsfig{file=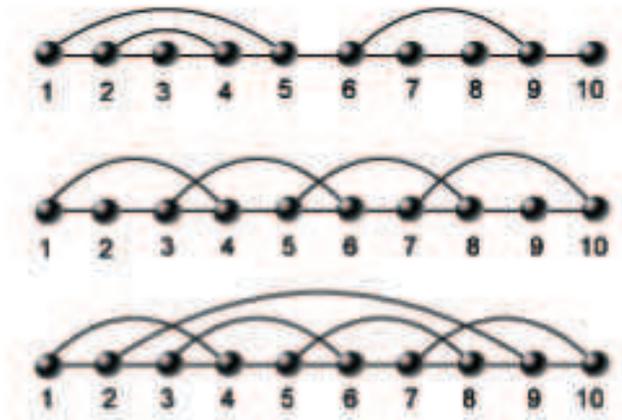,width=0.7\textwidth}\hskip15pt
 }
\caption{\small $k$-noncrossing RNA structures.
(a) secondary structure (with isolated labels $3,7,8,10$),
(b) bi-secondary structure, $2,9$ being isolated
(c) $3$-noncrossing structure, which is
    {\it not} a bi-secondary structure In fact, this is {\it the} smallest
    $3$-noncrossing RNA structure which is not a bi-secondary structure.
}
\label{F:4}
\end{figure}
%%%
%%%%%%%%%%%%%%%%%%%%%%%%%%%%%%%%%%%%%%%%%%%%%%%%%%%%%%%%%%%%%%%%%%%%%%%%%
%%%

\subsection{Organization and main results}
In Section~\ref{S:pre} provide the necessary background on
$k$-noncrossing RNA structures and the generating function
$\sum_{n\ge 0}{\sf S}_k(n)z^n$. In Section~\ref{S:exp} we derive the
exponential factor for $k$-noncrossing RNA structures, i.e.~we
compute the base at which $k$-noncrossing RNA structures
asymptotically grow. The exponential factor is {\it the} key result
for all complexity considerations arising in the context of
prediction algorithms for RNA pseudoknot structures. To make it
easily accessible to a broad readership we give an elementary proof
based on real analysis and transformations of the generating
function. Central to our proof is a functional identity
(Lemma~\ref{L:func}) whose true power is revealed only later in
Section~\ref{S:sub}, where it is put in the context of analytic
functions. Remarkably, Stadler's upper bound for bi-secondary
structures coincides with the exact exponential factor obtained via
Theorem~\ref{T:asy1} for $3$-noncrossing RNA structures up to
$O(10^{-2})$. In Section~\ref{S:sub} we compute the asymptotics for
$2$-noncrossing RNA structures and $3$-noncrossing RNA structures,
respectively. Since the method via implicit functions used for
secondary structures \cite{Schuster:98a} does not work for $k>2$ we
develop a new approach which is based on concepts developed by
Flajolet {\it et.al.} using singular expansions and transfer
theorems
\cite{Flajolet:05,Flajolet:99,Flajolet:94,Popken:53,Odlyzko:92}. The
basic strategy is as follows: we first obtain an analytic
continuation $f(z)$ generalizing the functional equation of
Lemma~\ref{L:func} to complex indeterminant $z$. For $k=3$ we obtain
an expression involving the Legendre polynomial
$P_{\frac{3}{2}}^{-1}(z)$ indicating that the type of singularity is
fundamentally different from the branch-point singularity of the
square root. In Figure~\ref{F:5} we display the analytic
continuation of $\sum_{n\ge 0}{\sf S}_3(n)z^n$ at the dominant
singularity, $\rho_3=\frac{5-\sqrt{21}}{2}$ and its singular
expansion.
%%%%
%%%%%%%%%%%%%%%%%%%%%%%%%%%%%%%%%%%%%%%%%%%%%%%%%%%%%%%%%%%%%%%%%%%%%%%%%%
%%%%
\begin{figure}[ht]
\centerline{%
\epsfig{file=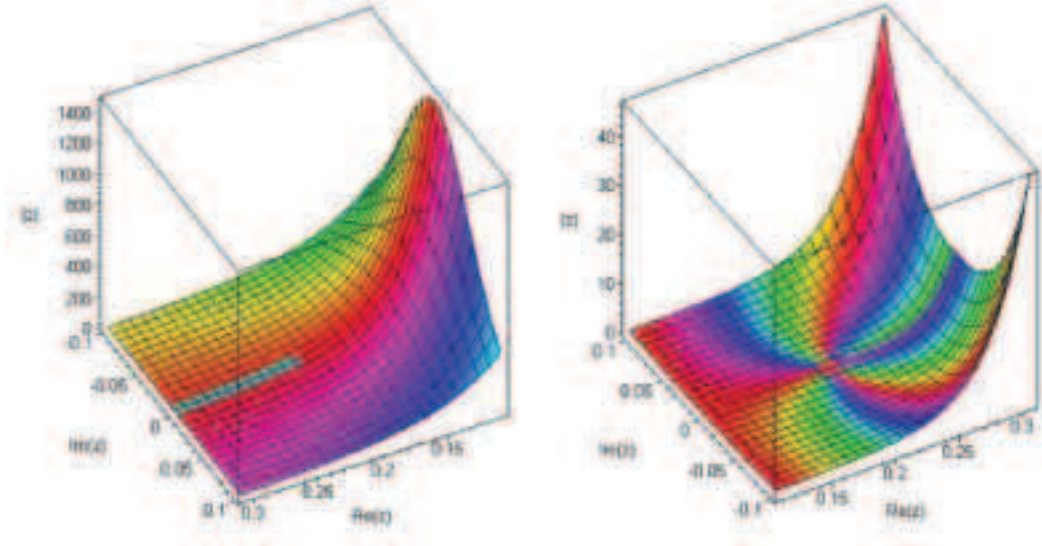,width=1.0\textwidth}\hskip15pt
 }
\caption{\small Toroidal harmonics and its singular expansion. We display
the analytic continuation of $\sum_{n\ge 0}{\sf S}_3(n)z^n$, the
generating function of $3$-noncrossing RNA structures (left) and its singular
expansion (right) at the dominant singularity
$\rho_3=\frac{5-\sqrt{21}}{2}$.}
\label{F:5}
\end{figure}
%%%%
%%%%%%%%%%%%%%%%%%%%%%%%%%%%%%%%%%%%%%%%%%%%%%%%%%%%%%%%%%%%%%%%%%%%%%%%%%
%%%%
We proceed by proving that $f(z)$ the dominant singularity is indeed unique.
The next step is to establish that there exists a singular
expansion for $f(z)$, i.e.~there exists a function $h$ such that
$f(z)=O(h(z))$ at the dominant singularity (see Section~\ref{S:sub}).
Intuitively, this singular expansion approximates $f(z)$ well enough to
retrieve precise asymptotic expansions of the coefficients via transfer
theorems \cite{Flajolet:05,Gao:92,Handbook:95}.
The existence of the singular expansion can be deduced from the particular
form of the generating function for $k$-noncrossing RNA structures.
Due to Lemma~\ref{L:ana} it suffices to analyze the coefficients
$f_3(2n,0)$, which are known via the determinant formula of Bessel
functions in eq.~(\ref{E:ww1}). We then proceed using this particular
form of $f_3(2n,0)$ to explicitly compute the singular expansion and
show in the process how the logarithmic term
arises naturally from elementary calculations.
It should be remarked that we use the transfer theorems since our
generating function is the composition of two analytic functions
$f(\vartheta(z))$. We then show that the type of the singualrity
of $f(\vartheta(z))$ coincides with the type of singularity of the
function $f(z)$. The phenomenon of the persistence of the singularity of
the ``outer'' function $f(z)$ is known as the {\it supercritical case}
\cite{Flajolet:05}. This will allow us to obtain the asymptotics of
the coefficients of the function $f(\vartheta(z))$.
One main result of the paper is the formula
\begin{eqnarray}
{\sf S}_3(n) & \sim & \frac{10.4724\, 4!}{n(n-1)\dots(n-4)}\,
\left(\frac{5+\sqrt{21}}{2}\right)^n \ .
\end{eqnarray}
In order to assess the quality of our formula, let us list
the sub exponential factors for $k=2$ and $k=3$, obtained from
Theorem~\ref{T:asy2} and
Theorem~\ref{T:asy3}:
\begin{eqnarray*}
{\sf s}_{2}^{}(n) & \sim & 1.1002\,
\left[\frac{1}{n^{\frac{3}{2}}}-\frac{7}
{8n^{\frac{5}{2}}}-\frac{111}{128n^{\frac{7}{2}}}+
\frac{893}{1024n^{\frac{9}{2}}}+{O}(n^{-\frac{11}{2}})\right] \\
{\sf s}_{3}^{}(n) & \sim & \frac{10.4724\cdot
4!}{n(n-1)\dots(n-4)}\, \sim 251.3375\left[\frac{1}{
n^5}-\frac{35}{4n^6}+\frac{1525}{32 n^7}+{O}(n^{-8})\right] \ .
\end{eqnarray*}
In the table below we list the sub exponential factors, i.e.~we compare
for $k=2,3$ the quantities ${\sf S}_{k}(n)/(\frac{3+\sqrt{5}}{2})^n$ and
${\sf s}_{k}^{}(n)$, respectively. ${\sf S}_2(n)$ and ${\sf S}_3(n)$ are
given by the generating function of $k$-noncrossing RNA structures.
\begin{center}
\begin{tabular}{c|c|c|c|c}
\hline
  \multicolumn{5}{c}{\textbf{The sub exponential factor}}\\
  \hline
$n$ & ${\sf S}_{2}(n)/(\frac{3+\sqrt{5}}{2})^n$ & ${\sf s}_{2}^{}(n)$
& ${\sf S}_{3}(n)/(\frac{5+\sqrt{21}}{2})^n$ & ${\sf s}_{3}^{}(n)$\\
\hline \small 10 & \small $2.796\times10^{-2}$ & \small $3.124\times10^{-2}$ & \small $5.229\times 10^{-4}$ & \small$1.512\times 10^{-3}$\\
\small 20  & \small $1.100\times 10^{-2}$ & \small $1.164\times10^{-2}$ & \small $3.358\times 10^{-5}$ & \small$5.354\times 10^{-5}$\\
\small 30  & \small $6.215\times 10^{-3}$ & \small $6.452\times10^{-3}$ & \small $5.776\times 10^{-6}$ & \small$7.874\times 10^{-6}$\\
\small 40  & \small $4.114\times 10^{-3}$ & \small $4.229\times10^{-3}$ & \small $1.576\times 10^{-6}$ & \small$1.991\times 10^{-6}$\\
\small 50  & \small $2.980\times 10^{-3}$ & \small $3.043\times10^{-3}$ & \small $5.627\times 10^{-7}$ & \small$6.789\times 10^{-7}$\\
\small 60  & \small $2.284\times 10^{-3}$ & \small $2.324\times10^{-3}$ & \small $2.397\times 10^{-7}$ & \small$2.804\times 10^{-7}$\\
\small 70  & \small $1.822\times 10^{-3}$ & \small $1.849\times10^{-3}$ & \small $1.156\times 10^{-7}$ & \small$1.323\times 10^{-7}$\\
\small 80  & \small $1.500\times 10^{-3}$ & \small $1.516\times10^{-3}$ & \small $6.123\times 10^{-8}$ & \small$6.888\times 10^{-8}$\\
\small 90  & \small $1.259\times 10^{-3}$ & \small $1.273\times10^{-3}$ & \small $3.483\times 10^{-8}$ & \small$3.868\times 10^{-8}$\\
\small 100  & \small $1.078\times 10^{-3}$ & \small $1.088\times10^{-3}$ & \small $2.098\times 10^{-8}$ & \small$2.305\times 10^{-8}$\\
\small 1000 & \small $3.484\times 10^{-5}$ & \small $3.475 \times 10^{-5}$ & \small $2.475\times 10^{-13}$ & \small $2.492\times 10^{-13}$ \\
\small 10000 & \small $1.104 \times 10^{-6}$ &\small $1.100\times 10^{-6}$ & \small $2.517 \times 10^{-18}$ & \small $2.516 \times 10^{-18}$
%\small 20000 & \small $3.904\times 10^{-7}$ &\small $3.904\times 10^{-7}$ & \small $7.874 \times 10^{-20}$ & \small $7.863 \times 10^{-20}$
\end{tabular}
\end{center}

%%%
%%%%%%%%%%%%%%%%%%%%%%%%%%%%%%%%%%%%%%%%%%%%%%%%%%%%%%%%%%%%%%%%%%%%%%%%
%%%

\section{$k$-noncrossing RNA structures}\label{S:pre}

%%%
%%%%%%%%%%%%%%%%%%%%%%%%%%%%%%%%%%%%%%%%%%%%%%%%%%%%%%%%%%%%%%%%%%%%%%%%
%%%
Suppose we are given the primary RNA sequence
$$
{\bf A}{\bf A}{\bf C}{\bf C}{\bf A}{\bf U}{\bf G}{\bf U}{\bf G}{\bf G}
{\bf U}{\bf A}{\bf C}{\bf U}{\bf U}{\bf G}{\bf A}{\bf U}{\bf G}{\bf G}
{\bf C}{\bf G}{\bf A}{\bf C}  \ .
$$
We then identify an RNA structure with the set of all bonds
different from the backbone-bonds of its primary sequence, i.e.~the
arcs $(i,i+1)$ for $1\le i\le n-1$. Accordingly an RNA structure is
a combinatorial graph over the labels of the nucleotides of the
primary sequence. These graphs can be represented in several ways.
In Figure~\ref{F:6} we represent a particular structure with
loop-loop interactions in two ways.
%%%%
%%%%%%%%%%%%%%%%%%%%%%%%%%%%%%%%%%%%%%%%%%%%%%%%%%%%%%%%%%%%%%%%%%%%%%%%%%
%%%%
\begin{figure}[ht]
\centerline{%
\epsfig{file=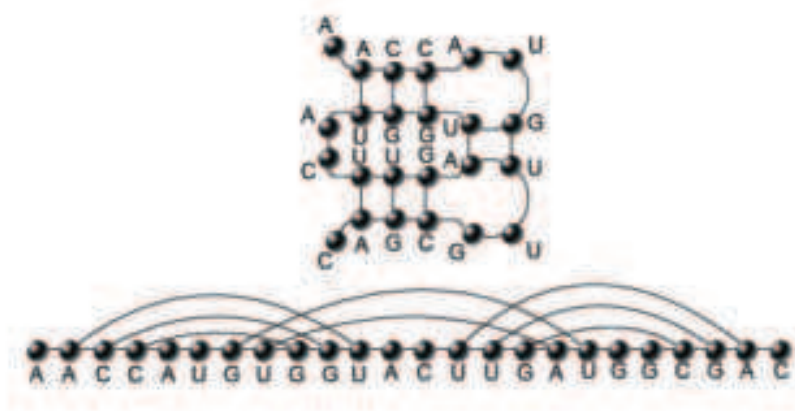,width=0.8\textwidth}\hskip15pt }
\caption{\small
A $3$-noncrossing RNA structure, as a planar graphs (top) and as a
diagram (bottom)} \label{F:6}
\end{figure}
%%%%
%%%%%%%%%%%%%%%%%%%%%%%%%%%%%%%%%%%%%%%%%%%%%%%%%%%%%%%%%%%%%%%%%%%%%%%%%%
%%%%
In the following we will consider structures as diagram
representations of digraphs. A digraph $D_n$ is a pair of sets
$V_{D_n},E_{D_n}$, where $V_{D_n}=
\{1,\dots,n\}$ and $E_{D_n}\subset \{(i,j)\mid 1\le i< j\le n\}$.
$V_{D_n}$ and $E_{D_n}$ are called vertex and arc set, respectively.
A $k$-noncrossing digraph is a digraph in which all
vertices have degree $\le 1$ and which does not contain a $k$-set of
arcs that are mutually intersecting, or more formally
\begin{eqnarray}
\ \not\exists\,
(i_{r_1},j_{r_1}),(i_{r_2},j_{r_2}),\dots,(i_{r_k},j_{r_k});\quad & &
i_{r_1}<i_{r_2}<\dots<i_{r_k}<j_{r_1}<j_{r_2}<\dots<j_{r_k} \ .
\end{eqnarray}
We will represent digraphs as a diagrams (Figure~\ref{F:6}) by
representing the vertices as integers on a line and connecting any
two adjacent vertices by an arc in the upper-half plane. The direction
of the arcs is implicit in the linear ordering of the vertices and
accordingly omitted.
%%%
%%%%%%%%%%%%%%%%%%%%%%%%%%%%%%%%%%%%%%%%%%%%%%%%%%%%%%%%%%%%%%%%%%%%%%
%%%
\begin{definition}\label{D:rna}
An $k$-noncrossing RNA structure is a digraph in which all vertices
have degree $\le 1$, that does not contain a $k$-set of mutually
intersecting arcs and $1$-arcs, i.e.~arcs of the form $(i,i+1)$,
respectively. We denote the number of RNA structures by ${\sf
S}_k(n)$ and the number of RNA structures with exactly $\ell$
isolated vertices and with exactly $h$ arcs by ${\sf S}_k(n,\ell)$
and ${\sf S}_k'(n,h)$, respectively. Note that ${\sf S}_k'(n,h)=
{\sf S}_k(n,{n-2h})$. We call an RNA structure restricted if
and only if it does not contain any $2$-arcs, i.e.~an arc of the
form $(i,i+2)$.
\end{definition}
%%%
%%%%%%%%%%%%%%%%%%%%%%%%%%%%%%%%%%%%%%%%%%%%%%%%%%%%%%%%%%%%%%%%%%%%%%
Let $f_{k}(n,\ell)$ denote the number of $k$-noncrossing digraphs with
$\ell$ isolated points. We have shown in \cite{Reidys:07pseu} that
\begin{align}\label{E:ww0}
f_{k}(n,\ell)& ={n \choose \ell} f_{k}(n-\ell,0) \\
\label{E:ww1}
\det[I_{i-j}(2x)-I_{i+j}(2x)]|_{i,j=1}^{k-1} &=
\sum_{n\ge 1} f_{k}(n,0)\,\frac{x^{n}}{n!} \\
\label{E:ww2}
e^{x}\det[I_{i-j}(2x)-I_{i+j}(2x)]|_{i,j=1}^{k-1}
&=(\sum_{\ell \ge 0}\frac{x^{\ell}}{\ell!})(\sum_{n \ge
1}f_{k}(n,0)\frac{x^{n}}{n!})=\sum_{n\ge 1}
\left\{\sum_{\ell=0}^nf_{k}(n,\ell)\right\}\,\frac{x^{n}}{n!} \ .
\end{align}
In particular we obtain for
$k=2$ and $k=3$
\begin{equation}\label{E:2-3}
f_2(n,\ell)  =  \binom{n}{\ell}\,C_{(n-\ell)/2}\quad
\text{\rm and}\quad  f_{3}(n,\ell)=
{n \choose \ell}\left[C_{\frac{n-\ell}{2}+2}C_{\frac{n-\ell}{2}}-
      C_{\frac{n-\ell}{2}+1}^{2}\right] \ ,
\end{equation}
where $C_m=\frac{1}{m+1}\binom{2m}{m}$ is the $m$th Catalan number.
The derivation of the generating function of $k$-noncrossing RNA structures,
given in Theorem~\ref{T:cool1} below uses advanced methods and novel
constructions of enumerative combinatorics due to Chen~{\it et.al.}
\cite{Chen:07a,Gessel:92a} and Stanley's mapping between matchings and
oscillating tableaux i.e.~families of Young diagrams in which any two
consecutive shapes differ by exactly one square.
The enumeration is obtained using the
reflection principle due to Gessel and Zeilberger \cite{Gessel:92a} and
Lindstr\"om \cite{Lindstroem:73a} combined with an inclusion-exclusion
argument in order to eliminate the arcs of length $1$. In
\cite{Reidys:07pseu} generalizations to restricted (i.e. where arcs of
the form $(i,i+2)$ are excluded) and circular RNA structures are given.
%%%
%%%%%%%%%%%%%%%%%%%%%%%%%%%%%%%%%%%%%%%%%%%%%%%%%%%%%%%%%%%%%%%%%%%%%%%%%
%%%
\begin{theorem}\label{T:cool1}\cite{Reidys:07pseu}
Let $k\in\mathbb{N}$, $k\ge 2$, then the
number of RNA structures with $\ell$ isolated vertices,
${\sf S}_k(n,\ell)$, is given by
\begin{equation}\label{E:da}
{\sf S}_k(n,\ell) = \sum_{b=0}^{(n-\ell)/2}
                   (-1)^b\binom{n-b}{b}f_k(n-2b,\ell)  \  ,
\end{equation}
where $f_k(n-2b,\ell)$ is given by the generating function in
eq.~{\rm (\ref{E:ww1})}.
Furthermore the number of $k$-noncrossing RNA structures, ${\sf S}_k(n)$
is
\begin{equation}\label{E:sum}
{\sf S}_k(n)
=\sum_{b=0}^{\lfloor n/2\rfloor}(-1)^{b}{n-b \choose b}
\left\{\sum_{\ell=0}^{n-2b}f_{k}(n-2b,\ell)\right\}
\end{equation}
where $\{\sum_{\ell=0}^{n-2b}f_{k}(n-2b,\ell)\}$ is given by the
generating function in eq.~{\rm (\ref{E:ww2})}.
\end{theorem}
%%%
%%%%%%%%%%%%%%%%%%%%%%%%%%%%%%%%%%%%%%%%%%%%%%%%%%%%%%%%%%%%%%%%%%%%%%%%%
%%%

%%%
%%%%%%%%%%%%%%%%%%%%%%%%%%%%%%%%%%%%%%%%%%%%%%%%%%%%%%%%%%%%%%%%%%%%%%%%%
%%%

\section{The exponential factor}\label{S:exp}

%%%
%%%%%%%%%%%%%%%%%%%%%%%%%%%%%%%%%%%%%%%%%%%%%%%%%%%%%%%%%%%%%%%%%%%%%%%%%
%%%
In this section we obtain the exponential growth factor of the
coefficients ${\sf S}_k(n)$. Let us begin by considering the
generating function $\sum_{n\ge 0}{\sf S}_k(n)z^n$ as a power series
over $\mathbb{R}$. Since $\sum_{n\ge 0}{\sf S}_k(n)z^n$ has
monotonously increasing coefficients $\lim_{n\to\infty}{\sf
S}_k(n)^{\frac{1}{n}}$ exists and determines via Hadamard's formula
its radius of convergence. As we already mentioned, due to the
inclusion-exclusion form of the terms ${\sf S}_k(n)$, it is not
obvious however, how to compute this radius of convergence. Our
strategy consists in first showing that ${\sf S}_k(n)$ is closely
related to $f_k(2n,0)$ via a functional relation of generating functions.
%%%
%%%%%%%%%%%%%%%%%%%%%%%%%%%%%%%%%%%%%%%%%%%%%%%%%%%%%%%%%%%%%%%%%%%%%%%%%
%%%
\begin{lemma}\label{L:func}
Let $z$ be an indeterminant over $\mathbb{R}$ and $w\in\mathbb{R}$ a
parameter. Let furthermore $\rho_k(w)$ denote the radius of convergence of
the power series $\sum_{n\ge 0} [\sum_{h\le n/2} {\sf S}_k(n,h) w^{2h}]
z^n$. Then for $\vert z\vert < \rho_k(w)$ holds
\begin{equation}\label{E:rr}
\sum_{n\ge 0} \sum_{h\le n/2} {\sf S}_k(n,h) w^{2h} z^n =
\frac{1}{w^2z^2-z+1}
\sum_{n\ge 0} f_k(2n,0) \left(\frac{wz}{w^2z^2-z+1}\right)^{2n} \ .
\end{equation}
In particular we have for $w=1$,
\begin{equation}\label{E:oha}
\sum_{n\ge 0} {\sf S}_k(n) z^{n+1} =
\sum_{n\ge 0} f_k(2n,0) \left(\frac{z}{z^2-z+1}\right)^{2n+1} \ .
\end{equation}
\end{lemma}
%%%
%%%%%%%%%%%%%%%%%%%%%%%%%%%%%%%%%%%%%%%%%%%%%%%%%%%%%%%%%%%%%%%%%%%%%%%%%
%%%
The proof of Lemma~\ref{L:func} is a bit technical and consists in a series
of changes of orders of summations and Laplace transforms. We give the proof
in Section~\ref{S:proofs}. In Section~\ref{S:sub} we will employ basic
complex
analysis and extend eq.~(\ref{E:rr}) to complex $z$. Lemma~\ref{L:func} is
the key to prove Theorem~\ref{T:asy1} below, where we obtain the exponential
factor for any $k>1$. In its proof we recruit the Theorem of Pringsheim
\cite{Titmarsh:39} which asserts that a power series $\sum_{n\ge 0}a_nz^n$
with $a_n\ge 0$ has its radius of convergence as dominant (but not
necessarily unique) singularity.
%%%
%%%%%%%%%%%%%%%%%%%%%%%%%%%%%%%%%%%%%%%%%%%%%%%%%%%%%%%%%%%%%%%%%%%%%%%%%
%%%
\begin{theorem}\label{T:asy1}
Let $k$ be a positive integer, $k>1$ and let $r_k$ be the radius of
convergence of the power series $\sum_{n\ge 0}f_k(2n,0)z^{2n}$. Then
the power series $\sum_{n\ge 0}{\sf S}_k(n)z^n$ has the real valued,
dominant singularity at $ \rho_k=
\frac{1+\frac{1}{{r_{k}}}}{2}-\sqrt{\left(\frac{1+\frac{1}{{r_{k}}}}
{2}\right)^2-1} $ and for the number of $k$-noncrossing RNA
structures holds
\begin{equation}\label{E:rel}
{\sf S}_k(n)\sim \left(\frac{1}{\rho_k}\right)^n \ .
\end{equation}
\end{theorem}
%%%
%%%%%%%%%%%%%%%%%%%%%%%%%%%%%%%%%%%%%%%%%%%%%%%%%%%%%%%%%%%%%%%%%%%%%%%%%
%%%
We will prove later in  Theorem~\ref{T:asy2} and Theorem~\ref{T:asy3} that
for $k=2$ and $k=3$ the dominant singularities $\rho_2$ and $\rho_3$ are
unique, respectively.
\begin{proof}
Suppose we are given $r_k$, then $r_k\le \frac{1}{2}$ (this follows
immediately from $C_n\sim 2^{2n}$ via Stirling's formula) and
obviously, $(z-\frac{1}{2})^2+\frac{3}{4}$ has no roots over
$\mathbb{R}$. The functional identity of Lemma~\ref{L:func} allows
us to derive the radius of convergence of $\sum_{n\ge 0}S_k(n)z^n$.
Setting $w=1$ Lemma~\ref{L:func} yields
\begin{equation}\label{E:w=1}
\sum_{n\ge 0} {\sf S}_k(n) z^n =
\frac{1}{(z-\frac{1}{2})^2+\frac{3}{4}}
\sum_{n\ge 0} f_k(2n,0) \left(\frac{z}{(z-\frac{1}{2})^2+\frac{3}{4}}\right)
^{2n} \ .
\end{equation}
$f_k(2n,0)$ is monotone, whence the limit $\lim_{n\to
\infty}f_k(2n,0)^{ \frac{1}{2n}}$ exists and applying Hadamard's
formula: $\lim_{n\to
\infty}f_k(2n,0)^{\frac{1}{2n}}=\frac{1}{{r_{k}}}$. For $z\in
\mathbb{R}$, we proceed by computing the roots of
$\left|\frac{z}{z^2-z+1}\right|={r_{k}}$ which for $r_k\le
\frac{1}{2}$ has the minimal root $\rho_k=
\frac{1+\frac{1}{{r_{k}}}}{2}-\left(
\left(\frac{1+\frac{1}{{r_{k}}}}{2}\right)^2-1\right)^{\frac{1}{2}}$.
We next show that $\rho_k$ is indeed the radius of convergence of
$\sum_{n\ge 0} {\sf S}_k(n) z^n$. For this purpose we observe that
the map
\begin{equation}
\vartheta\colon [0, \frac{1}{2}]\longrightarrow [0,\frac{2}{3}],
\quad z\mapsto \frac{z}{(z-\frac{1}{2})^2+\frac{3}{4}} , \qquad
\text{\rm where} \quad\vartheta(\rho_k)={r_k}
\end{equation}
is bijective, continuous and strictly increasing.
Continuity and strict monotonicity of $\vartheta$ guarantee in view of
eq.~(\ref{E:w=1}) that $\rho_k$, is indeed the radius of convergence of
the power series $\sum_{n\ge 0} {\sf S}_k(n) z^n$. In order to show that
$\rho_k$ is a dominant singularity we consider
$\sum_{n\ge 0}{\sf S}_k(n)z^n$ as a power series over $\mathbb{C}$.
Since ${\sf S}_k(n)\ge 0$, the theorem of Pringsheim~\cite{Titmarsh:39}
guarantees that $\rho_k$ itself is a singularity. By construction
$\rho_k$ has minimal absolute value and is accordingly dominant.
Since ${\sf S}_k(n)$ is monotone $\lim_{n\to\infty}{\sf S}_k(n)^{\frac{1}
{n}}$ exists and we obtain using Hadamard's formula
\begin{equation}
\lim_{n\to\infty}{\sf S}_k(n)^{\frac{1}{n}}=\frac{1}{\rho_k},\quad
\text{\rm or equivalently}\quad {\sf S}_k(n)\sim \left(\frac{1}{\rho_k}
\right)^n \, ,
\end{equation}
from which eq.~(\ref{E:rel}) follows and the proof of the theorem is complete.
\end{proof}
%%%%
%%%%%%%%%%%%%%%%%%%%%%%%%%%%%%%%%%%%%%%%%%%%%%%%%%%%%%%%%%%%%%%%%%%%%%%%%
%%%
\section{Asymptotics of $3$-noncrossing RNA structures}\label{S:sub}
%%%
%%%%%%%%%%%%%%%%%%%%%%%%%%%%%%%%%%%%%%%%%%%%%%%%%%%%%%%%%%%%%%%%%%%%%%%%%
%%%
In this section we provide the asymptotics for RNA secondary and
$3$-noncrossing RNA structures. For $k=2$ and $k=3$, i.e.~for RNA
secondary and $3$-noncrossing RNA structures, respectively we will
explicitly obtain analytic continuations of the power series
$\sum_{n\ge 0}{\sf S}_2(n)z^n$ and $\sum_{n\ge 0}{\sf S}_3(n)z^n$,
respectively. As a result the dominant singularity relevant for the
asymptotics is known and Theorem~\ref{T:asy1} becomes obsolete.
However, it is not entirely trivial to derive the analytic
continuations for arbitrary crossing numbers $k$. In the context of
complexity of prediction algorithms for RNA pseudoknot structures it
suffices to obtain the exponential factor which is given via
Theorem~\ref{T:asy1}.

We begin by revealing the ``true'' power of Lemma~\ref{L:func} in the
context of analytic functions.
%%%
%%%%%%%%%%%%%%%%%%%%%%%%%%%%%%%%%%%%%%%%%%%%%%%%%%%%%%%%%%%%%%%%%%%%%%%%%
%%%
\begin{lemma}\label{L:ana}
Let $k\ge 1$ be an integer, then we have for arbitrary $z\in\mathbb{C}$
with the property $\vert z\vert <\rho_k$ the equality
\begin{equation}\label{E:rr2}
\sum_{n\ge 0} {\sf S}_k(n) z^{n}= \frac{z}{z^2-z+1}
\sum_{n\ge 0} f_k(2n,0) \left(\frac{z}{z^2-z+1}\right)^{2n} \ .
\end{equation}

\end{lemma}
%%%
%%%%%%%%%%%%%%%%%%%%%%%%%%%%%%%%%%%%%%%%%%%%%%%%%%%%%%%%%%%%%%%%%%%%%%%%%
%%%
\begin{proof}
The power series $\sum_{n\ge 0} {\sf S}_k(n) z^{n}$ and
$\sum_{n\ge 0} f_k(2n,0) \left(\frac{z}{z^2-z+1}\right)^{2n}$
are analytic in a disc of radius $0<\epsilon<\rho_k$ and
according to Lemma~\ref{L:func} coincide on the interval
$]-\epsilon,\epsilon [$. Therefore both functions are equal
on the sequence $(\frac{1}{n})_{n\in\mathbb{N}}$
which converges to $0$ and
standard results of complex analysis (zeros of nontrivial
analytic functions are isolated) imply that eq.~(\ref{E:rr2})
holds for any $z\in\mathbb{C}$ with $\vert z\vert<\rho_k$, whence
the lemma.
\end{proof}

The derivation of the sub exponential factors is based on singular expansions
in combination with a transfer theorem, which recruits Hankel contours,
see Figure~\ref{F:7}. Let us begin by specifying a suitable domain for our
Hankel contours tailored for Theorem~\ref{T:transfer1}.
%%%
%%%%%%%%%%%%%%%%%%%%%%%%%%%%%%%%%%%%%%%%%%%%%%%%%%%%%%%%%%%%%%%%%%%%%%%%%
%%%
\begin{definition}\label{D:delta}
Given two numbers $\phi,R$, where $R>1$ and $0<\phi<\frac{\pi}{2}$ and
$\rho\in\mathbb{R}$ the open domain $\Delta_\rho(\phi,R)$ is defined as
\begin{equation}
\Delta_\rho(\phi,R)=\{ z\mid \vert z\vert < R, z\neq \rho,\,
\vert {\rm Arg}(z-\rho)\vert >\phi\}
\end{equation}
A domain is a $\Delta_\rho$-domain if it is of the form
$\Delta_\rho(\phi,R)$ for some $R$ and $\phi$.
A function is $\Delta_\rho$-analytic if it is analytic in some
$\Delta_\rho$-domain.
\end{definition}
%%%
%%%%%%%%%%%%%%%%%%%%%%%%%%%%%%%%%%%%%%%%%%%%%%%%%%%%%%%%%%%%%%%%%%%%%%%%%
%%%
%%%%
%%%%%%%%%%%%%%%%%%%%%%%%%%%%%%%%%%%%%%%%%%%%%%%%%%%%%%%%%%%%%%%%%%%%%%%%%%
%%%%
\begin{figure}[ht]
\centerline{%
\epsfig{file=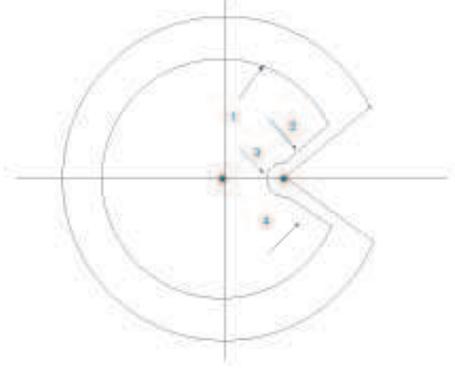,width=0.5\textwidth}\hskip15pt
 }
\caption{\small $\Delta_1$-domain enclosing a Hankel contour. We
assume $z=1$ to be the unique dominant singularity. The coefficients
are obtained via Cauchy's integral formula and the integral path is
decomposed in $4$ segments. Segment $1$ becomes asymptotically
irrelevant since by construction the function involved is bounded on
this segment. Relevant are the rectilinear segments $2$ and $4$ and
the inner circle $3$. The only contributions to the contour integral
are being made here, which shows why the singular expansion allows
to approximate the coefficients so well.} \label{F:7}
\end{figure}
%%%%
%%%%%%%%%%%%%%%%%%%%%%%%%%%%%%%%%%%%%%%%%%%%%%%%%%%%%%%%%%%%%%%%%%%%%%%%%%
%%%%
Since the Taylor coefficients have the property
\begin{equation}\label{E:scaling}
\forall \,\gamma\in\mathbb{C}\setminus 0;\quad
[z^n]f(z)=\gamma^n [z^n]f(\frac{z}{\gamma}) \ ,
\end{equation}
we can, w.l.o.g.~reduce our analysis to the case where $1$ is the dominant
singularity. We use the notation
\begin{equation}\label{E:genau}
\left(f(z)=O\left(g(z)\right) \
\text{\rm as $z\rightarrow \rho$}\right)\quad \Longleftrightarrow \quad
\left(f(z)/g(z) \ \text{\rm is bounded as $z\rightarrow \rho$}\right) \ .
\end{equation}
and if we write $f(z)=O(g(z))$ it is implicitly assumed that $z$
tends to a (unique) singularity. $[z^n]\,f(z)$ denotes the
coefficient of $z^n$ in the power series expansion of $f(z)$ around
$0$.
%%%
%%%%%%%%%%%%%%%%%%%%%%%%%%%%%%%%%%%%%%%%%%%%%%%%%%%%%%%%%%%%%%%%%%%%%%%%%
%%%
\begin{theorem}\label{T:transfer1}\cite{Flajolet:05}
Let $\alpha$ be an arbitrary complex number in $\mathbb{C}\setminus
\mathbb{Z}_{\le 0}$ and suppose $f(z)=O((1-z)^{-\alpha})$, then
\begin{eqnarray*}
[z^n]\, f(z) & \sim & K\ n^{\alpha-1}\left[
1+\frac{\alpha(\alpha-1)}{2n}+\frac{\alpha(\alpha-1)(\alpha-2)(3\alpha-1)}
{24 n^2}+\right. \\
&&\qquad  \quad
\left.\frac{\alpha^2(\alpha-1)^2(\alpha-2)(\alpha-3)}{48n^3}+\dots\right]
\quad \text{\it for some $K>0$.}
\end{eqnarray*}
Suppose $r\in\mathbb{Z}_{\ge 0}$, and
$f(z)=O((1-z)^{r}\ln^{}(\frac{1}{1-z}))$, then we have
\begin{equation}
[z^n]f(z)= K\,  (-1)^r\frac{r!}{n(n-1)\dots(n-r)} \quad \text{\it for some
$K>0$}\ .
\end{equation}
\end{theorem}
%%%
%%%%%%%%%%%%%%%%%%%%%%%%%%%%%%%%%%%%%%%%%%%%%%%%%%%%%%%%%%%%%%%%%%%%%%%%%
%%%
Let us first analyze the case $k=2$, which illustrates the general strategy
without the technicality of establishing the existence of a suitable singular
expansion. Here the generating function itself can be used directly
(i.e.~is its own singular expansion).
Our particular proof, given in Section~\ref{S:proofs}, exercises the base
strategy used in the proof of Theorem~\ref{T:asy3}. In particular,
Theorem~\ref{T:asy2} improves on the quality of approximation providing a
sub exponential factor of higher order compared to \cite{Schuster:98a}.
%%%
%%%%%%%%%%%%%%%%%%%%%%%%%%%%%%%%%%%%%%%%%%%%%%%%%%%%%%%%%%%%%%%%%%%%%%%%%
%%%
\begin{theorem}\label{T:asy2}
The number of RNA secondary i.e.~$2$-noncrossing RNA structures is
asymptotically given by
\begin{eqnarray} \label{E:konk2}
{\sf S}_2(n)  \sim
\frac{1.1002}{\sqrt{n}}
\left(\frac{1}{n+1}-\frac{1}{8n(n+1)}+
\frac{1}{128n^{3}}+\frac{5}{1024n^{4}}
+{O}(n^{-5})\right) \left(\frac{3+\sqrt{5}}{2}\right)^n \ .
\end{eqnarray}
\end{theorem}
%%%
%%%%%%%%%%%%%%%%%%%%%%%%%%%%%%%%%%%%%%%%%%%%%%%%%%%%%%%%%%%%%%%%%%%%%%%%%
%%%

%%%%%%%%%%%%%%%%%%%%%%%%%%%%% 3-noncrossing %%%%%%%%%%%%%%%%%%%%%%%%%%%%%

We next analyze the $3$-noncrossing RNA structures. Here the situation
changes dramatically since it has to be shown that a suitable singular
expansion exists. We will prove this using the determinant formula
arising in the context of the exponential generating function of
$f_k(2n,0)$ given in eq.~(\ref{E:ww1}).

%%%
%%%%%%%%%%%%%%%%%%%%%%%%%%%%%%%%%%%%%%%%%%%%%%%%%%%%%%%%%%%%%%%%%%%%%%%%%
%%%
\begin{theorem}\label{T:asy3}
The number of $3$-noncrossing RNA structures is asymptotically given by
\begin{eqnarray*}
\label{E:konk3}
{\sf S}_3(n) & \sim & \frac{10.4724\cdot 4!}{n(n-1)\dots(n-4)}\,
\left(\frac{5+\sqrt{21}}{2}\right)^n \ .\\
\end{eqnarray*}
\end{theorem}
%%%
%%%%%%%%%%%%%%%%%%%%%%%%%%%%%%%%%%%%%%%%%%%%%%%%%%%%%%%%%%%%%%%%%%%%%%%%%
%%%
\begin{proof}
{\it Claim $1$.}
The dominant singularity $\rho_3$ of the power series
$\sum_{n\ge 0} {\sf S}_3(n) z^n$ is unique.\\
%%%%%%%%%%%%%%%%%%%%%%%%%%%%%%%%%%%%%%%%%%%
In order to prove Claim $1$ we use Lemma~\ref{L:ana}, according to
which the analytic function $\Xi_3(z)$ is the analytic continuation
of the power series $\sum_{n\ge 0} {\sf S}_3(n) z^n$. We proceed by
showing that $\Xi_3(z)$ has exactly $6$ singularities in
$\mathbb{C}$, all of which have distinct moduli. The first two
singularities are the roots of the quadratic polynomial
$P(z)=(z-\frac{1}{2})^2+\frac{3}{4}$, given by
$\alpha_1=\frac{1}{2}+i \frac{\sqrt{3}}{2}$ and
$\alpha_2=\frac{1}{2}-i\frac{\sqrt{3}}{2}$. Next we observe that the
power series $\sum_{n\ge 0} f_k(2n,0) y^{n}$ has the analytic
continuation $\Psi(y)$ (obtained by MAPLE sumtools) given by
\begin{equation}\label{E:psi}
\Psi(y)=
\frac{-(1-16y)^{\frac{3}{2}}  P_{3/2}^{-1}(-\frac{16y+1}{16y-1})}
{16\, {y}^{\frac{5}{2}}} \ ,
\end{equation}
where $P_{\nu}^{m}(x)$ denotes the Legendre Polynomial of the first kind
with the parameters $\nu=\frac{3}{2}$ and $m=-1$.
$\Psi(y)$ has one dominant singularity
at $y=\frac{1}{16}$, which in view of $\vartheta(z)=(\frac{z}{z^2-z+1})^2$
induces exactly $4$ singularities of
$
\Xi_3(z)=\frac{1}{z^2-z+1}\,
                    \Psi\left(\left(\frac{z}{z^2-z+1}\right)^2\right)
$.
Indeed, $\Psi(y^2)$ has the two singularities $\mathbb{C}$:
$\beta_1=\frac{1}{4}$ and $\beta_2=-\frac{1}{4}$ which produce for
$\Xi_3(z)$ the four singularities $\rho_3=\frac{5-\sqrt{21}}{2}$,
$\zeta_2=\frac{5+\sqrt{21}}{2}$, $\zeta_3=\frac{-3-\sqrt{5}}{2}$ and
$\zeta_4=\frac{-3+\sqrt{5}}{2}$. Therefore all $6$ singularities of
$\Xi_3(z)$ have distinct moduli and Claim
$1$ follows.\\
{\it Claim $2$: the singular expansion}. $\Psi(z)$ is
$\Delta_{\frac{1}{16}}(\phi,R)$-analytic and has the singular expansion
$(1-16z)^4\ln\left(\frac{1}{1-16z}\right)$.
\begin{equation}
\forall\, z\in\Delta_{\frac{1}{16}}(\phi,R);\quad
\Psi(z)={O}\left((1-16z)^4\ln\left(\frac{1}{1-16z}\right)\right) \ .
\end{equation}
First $\Delta_{\frac{1}{16}}(\phi,R)$-analyticity of the function
$(1-16z)^4\ln\left(\frac{1}{1-16z}\right)$ is obvious.
We proceed by proving that $(1-16z)^4\ln\left(\frac{1}{1-16z}\right)$ is
the singular expansion.
Using the notation of falling factorials
$(n-1)_4=(n-1)(n-2)(n-3)(n-4)$ we observe
$$
f_3(2n,0)=C_{n+2}C_{n}-C_{n+1}^2= \frac{1}{(n-1)_4}
\frac{12(n-1)_4(2n+1)}{(n+3)(n+1)^2(n+2)^2}\,
\binom{2n}{n}^2 \ .
$$
With this expression for $f_3(2n,0)$ we arrive at the formal identity
\begin{eqnarray*}
\sum_{n\ge 5}16^{-n}f_3(2n,0)z^n  & = &
O(\sum_{n\ge 5}
\left[16^{-n}\,\frac{1}{(n-1)_4}
\frac{12(n-1)_4(2n+1)}{(n+3)(n+1)^2(n+2)^2}\,
\binom{2n}{n}^2-\frac{4!}{(n-1)_4}\frac{1}{\pi}\frac{1}{n}\right]z^n \\
& & + \sum_{n\ge 5}\frac{4!}{(n-1)_4}\frac{1}{\pi}\frac{1}{n}z^n) \ ,
\end{eqnarray*}
where $f(z)=O(g(z))$ denotes that the limit $f(z)/g(z)$ is bounded
for $z\rightarrow 1$, eq.~(\ref{E:genau}). It is clear that
\begin{eqnarray*}
& & \lim_{z\to 1}(\sum_{n\ge 5}\left[16^{-n}\,\frac{1}{(n-1)_4}
\frac{12(n-1)_4(2n+1)}{(n+3)(n+1)^2(n+2)^2}\,
\binom{2n}{n}^2-\frac{4!}{(n-1)_4}\frac{1}{\pi}\frac{1}{n}\right]z^n)  \\
&= &
\sum_{n\ge 5} \left[16^{-n}\,\frac{1}{(n-1)_4}
\frac{12(n-1)_4(2n+1)}{(n+3)(n+1)^2(n+2)^2}\,
\binom{2n}{n}^2-\frac{4!}{(n-1)_4}\frac{1}{\pi}\frac{1}{n}\right]
 <\kappa
\end{eqnarray*}
for some $\kappa< 0.0784$. Therefore we can conclude
\begin{equation}
\sum_{n\ge 5}16^{-n}f_3(2n,0)z^n=
O(\sum_{n\ge 5}\frac{4!}{(n-1)_4}\frac{1}{\pi}\frac{1}{n}z^n) \ .
\end{equation}
We proceed by interpreting the power series on the rhs, observing
\begin{equation}
\forall\, n\ge 5\, ; \qquad
[z^n]\left((1-z)^4\,\ln\frac{1}{1-z}\right)=
\frac{4!}{(n-1)\dots (n-4)}\frac{1}{n} \, ,
\end{equation}
whence $\left((1-z)^4\,\ln\frac{1}{1-z}\right)$ is
the unique analytic continuation of $\sum_{n\ge 5}\frac{4!}{(n-1)_4}
\frac{1}{\pi}\frac{1}{n}z^n$.
Using the scaling property of Taylor coefficients
$[z^n]f(z)=\gamma^n [z^n]f(\frac{z}{\gamma})$ we obtain
\begin{equation}\label{E:isses}
\forall\, z\in\Delta_{\frac{1}{16}}(\phi,R);\quad
\Psi(z) =O\left((1-16z)^4\ln\left(\frac{1}{1-16z}\right)\right) \ .
\end{equation}
Therefore we have proved that $(1-16z)^{4}\ln^{}(\frac{1}{1-16z})$ is the
singular expansion of $\Psi(z)$ at $z=\frac{1}{16}$, whence Claim $2$.
Our last step consists in verifying that the type of the singularity
does not change when passing from $\Psi(z)$ to $\Xi_3(z)=\frac{1}{z^2-z+1}
\Psi((\frac{z}{z^2-z+1})^2)$. That is, we show that the singular expansion
is not affected by substituting $\vartheta(z)=(\frac{z}{z^2-z+1})^2$.
\\
{\it Claim $3$: the singularity persists}.
For $z\in \Delta_{\frac{1}{\rho_3}}(\phi,R)$ we have
$\Xi_3(z) ={O}\left((1-\frac{z}{\rho_3})^4\ln(\frac{1}{1-\frac{z}{\rho_3}})
\right)$.\\
To prove the claim we first observe that Claim $2$ and Lemma~\ref{L:ana}
imply
\begin{align*}
\Xi_3(z)
&=O\left(
\frac{1}{z^2-z+1}\,\left[\left(1-16(\frac{z}{z^2-z+1})^2\right)^4
\ln\frac{1}{\left(1-16(\frac{z}{z^2-z+1})^2\right)}\right]\right) \ .
\end{align*}
The Taylor expansion of $q(z)=1-16(\frac{z}{z^2-z+1})^2$ at $\rho_3$ is
given by $q(z)=\frac{\sqrt{21}}{5-\sqrt{21}}(\rho_3-z)+{O}(z-\rho_3)^2$ and
setting $\alpha=\frac{\sqrt{21}}{5-\sqrt{21}}$ we compute
\begin{align*}
\frac{1}{z^2-z+1}\,
\left[q(z)^4\ln\frac{1}
{q(z)}\right]
&=
\frac{(\alpha(\rho_3-z)+{O}(z-\rho_3)^2)^4\ln\frac{1}{\alpha(\rho_3-z)+{O}
(z-\rho_3)^2}}{(z-\rho_3)^2+(2\rho_3-1)(z-\rho_3)+\rho_3^2-\rho_3+1}\\
&=
\frac{\left([\alpha+O(z-\rho_3)](\rho_3-z)^4
\ln\frac{1}{[\alpha+O(z-\rho_3)](\rho_3-z)}
\right)}{O(z-\rho_3)+\rho_3^2-\rho_3+1}
\\
&={O}((\rho_3-z)^4\ln\frac{1}{\rho_3-z}) \ ,
\end{align*}
whence Claim $3$. Now we are in the position to employ
Theorem~\ref{T:transfer1}, and obtain for ${\sf S}_3(n)$
\begin{align*}
{\sf S}_{3}(n)&\sim K'\, [z^n]\left((\rho_3-z)^4\ln\frac{1}{\rho_3-z}
\right) \sim K'\, \frac{4!}
{n(n-1)\dots(n-4)}(\frac{1}{\rho_3})^n  \ .
\end{align*}
Of course $K'$ can be computed from Theorem~\ref{T:cool1},
explicitly $K'=10.4724$ and the proof of the Theorem
is complete.
\end{proof}

%%%
%%%%%%%%%%%%%%%%%%%%%%%%%%%%%%%%%%%%%%%%%%%%%%%%%%%%%%%%%%%%%%%%%%%%%%%%
%%%

\section{Proofs}\label{S:proofs}

%%%
%%%%%%%%%%%%%%%%%%%%%%%%%%%%%%%%%%%%%%%%%%%%%%%%%%%%%%%%%%%%%%%%%%%%%%%%
%%%

{\bf Proof of Lemma~\ref{L:func}.}
First we observe that for $z,w\in [-1,1]$ the term $w^2z^2-z+1$ is strictly
positive.
We set
\begin{equation}
F_k(z,w)=\sum_{n\ge 0} \sum_{h\le n/2}{\sf S}_k(n,h)w^{2h}z^n
\end{equation}
and compute
\begin{eqnarray*}
F_k(z,w) & = & \sum_{n\ge 0}\sum_{h\le n/2}\sum_{j=0}^h(-1)^j\binom{n-j}{j}
\binom{n-2j}{2(h-j)} f_k(2(h-j),0)w^{2h}z^n \\
& = & \sum_{n\ge 0}\sum_{j\le n/2}\sum_{h=j}^{n/2}(-1)^j\binom{n-j}{j}
\binom{n-2j}{2(h-j)}f_k(2(h-j),0)w^{2h}z^n \\
& = & \sum_{j\ge 0}\sum_{n\ge
2j}\sum_{h=j}^{n/2}(-1)^j\binom{n-j}{j}
\binom{n-2j}{2(h-j)}f_k(2(h-j),0)w^{2h}z^n \\
& = & \sum_{j\ge 0}(-1)^j\frac{(wz)^{2j}}{j!}\sum_{n\ge 2j}(n-j)!
\sum_{h=j}^{n/2}\binom{n-2j}{2(h-j)}f_k(2(h-j),0)\frac{w^{2(h-j)}}{(n-2j)!}
z^{n-2j} \ .
\end{eqnarray*}
We shift summation indices $n'=n-2j$ and $h'=h-j$ and derive for the rhs
the following expression
\begin{eqnarray*}
& = & \sum_{j\ge 0}(-1)^j\frac{(wz)^{2j}}{j!}\sum_{n'\ge 0}(n'+j)!
\sum_{h=j}^{n/2}\binom{n'}{2(h-j)}f_k(2(h-j),0)\frac{w^{2(h-j)}}{n'!}
z^{n-2j} \\
& = & \sum_{j\ge 0}(-1)^j\frac{(wz)^{2j}}{j!}\sum_{n'\ge 0}(n'+j)!
\left\{
\sum_{h'=0}^{n/2-j=n'/2}\binom{n'}{2h'}f_k(2h',0)w^{2h'}\right\}
\frac{z^{n'}}{n'!}
\end{eqnarray*}
The idea is now to interpret the term $\sum_{h'=0}^{n'/2}\binom{n'}
{2h'}f_k(2h',0)w^{2h'}\frac{z^n}{n!}$ as a product of the two power series
$e^z$ and $\sum_{n\ge 0}f_k(2n,0)\frac{(wz)^{2n}}{(2n)!}$:
\begin{eqnarray*}
\sum_{\ell\ge 0}\frac{z^\ell}{\ell!}\sum_{n\ge 0}f_k(2n,0)\frac{(wz)^{2n}}
{(2n)!}
& = &
\sum_{n'\ge 0} \sum_{2n+\ell=n'}
\left\{\frac{1}{\ell!}\frac{1}{(2n)!}
f_k(2n,0)w^{2n}\right\}z^{n'}\\
& = & \sum_{n'\ge 0}\left\{\sum_{n=0}^{n'/2}
\binom{n'}{2n}f_k(2n,0)w^{2n}\right\}\frac{z^{n'}}{n'!} \ .
\end{eqnarray*}
We set $\eta_{n'}=\left\{\sum_{h'=0}^{n'/2}\binom{n'}{2h'}f_k(2h',0)
w^{2h'}\right\} $. By assumption we have $\vert z\vert<\rho_k(w)$ and we next
derive, using the Laplace transformation and interchanging integration and
summation
\begin{equation}\label{E:on1}
\sum_{n'\ge 0}(n'+j)!\eta_n
\frac{z^{n'}}{n'!}
=  \int_{0}^{\infty}
\sum_{n'\ge 0}\eta_{n'} \frac{(zt)^{n'}}{n'!} t^je^{-t} dt \ .
\end{equation}
Since $\vert z\vert<\rho_k(w)$ the above transformation is valid and using
\begin{equation}
 \sum_{n'\ge 0}\left\{\sum_{n=0}^{n'/2}
\binom{n'}{2n}f_k(2n,0)w^{2n}\right\}\frac{z^{n'}}{n'!}=
\sum_{\ell\ge 0}\frac{z^\ell}{\ell!}\sum_{n\ge 0}f_k(2n,0)\frac{(wz)^{2n}}
{(2n)!}
\end{equation}
we accordingly obtain
\begin{eqnarray}
\sum_{n'\ge 0}\eta_{n'} \frac{(zt)^{n'}}{n'!} t^je^{-t} dt & = &
 \int_{0}^{\infty}e^{tz}\sum_{n\ge 0}f_k(2n,0)\frac{(wzt)^{2n}}{(2n)!}
t^je^{-t}dt \ .
\end{eqnarray}
The next step is to substitute the term $\sum_{n'\ge 0}(n'+j)!\eta_n
\frac{z^{n'}}{n'!}$ in eq.~(\ref{E:on1}), whence consequently
\begin{eqnarray*}
F_k(z,w) & = &  \sum_{j\ge 0}(-1)^j\frac{(wz)^{2j}}{j!}
 \int_{0}^{\infty}e^{tz}\sum_{n\ge 0}f_k(2n,0)\frac{(wzt)^{2n}}
{(2n)!}t^je^{-t}dt\\
& = &  \int_{0}^{\infty}\sum_{j\ge 0}(-1)^j\frac{(wz)^{2j}}{j!}
e^{tz}\sum_{n\ge 0}f_k(2n,0)\frac{(wzt)^{2n}}
{(2n)!}t^je^{-t}dt \ .
\end{eqnarray*}
The summation over the index $j$ is just an exponential function and we
derive
\begin{eqnarray*}
& = &\int_{0}^{\infty}
e^{-(w^2z^2-z+1)t} \sum_{n\ge 0}f_k(2n,0)\frac{(wzt)^{2n}}
{(2n)!}dt \\
& = &\int_{0}^{\infty}
e^{-(w^2z^2-z+1)t} \sum_{n\ge 0}f_k(2n,0)\frac{1}{(2n)!}
\left(\frac{wz}{w^2z^2-z+1}\right)^{2n} ((w^2z^2-z+1)t)^{2n}dt
\end{eqnarray*}
We proceed by transforming the integral introducing $u=(w^2z^2-z+1)t$,
i.e.~ $dt=(w^2z^2-z+1)^{-1}du$ and accordingly arrive at
\begin{eqnarray*}
F_k(z,w) & = &  \sum_{n\ge 0}f_k(2n,0)\frac{1}{(2n)!}
\left(\frac{wz}{w^2z^2-z+1}\right)^{2n}
\int_{0}^{\infty} e^{-(w^2z^2-z+1)t}  ((w^2z^2-z+1)t)^{2n}
dt \\
& = &  \sum_{n\ge 0}f_k(2n,0)\frac{1}{(2n)!}
\left(\frac{wz}{w^2z^2-z+1}\right)^{2n}
\frac{1}{w^2z^2-z+1} (2n)! \\
& = & \frac{1}{w^2z^2-z+1}  \sum_{n\ge 0}f_k(2n,0)
\left(\frac{wz}{w^2z^2-z+1}\right)^{2n} \ ,
\end{eqnarray*}
whence the lemma. $\square$

{\bf Proof of Theorem~\ref{T:asy2}.} We shall begin by deriving the
asymptotics of $f_2(2n,0)=C_{n}$. Since $\sum_{n\ge
0}\binom{2n}{n}z^n=(1-4z)^{-\frac{1}{2}}$, we observe
$$
C_n =\frac{1}{n+1} [z^n]\, (1-4z)^{-\frac{1}{2}}
$$
and according to Theorem~\ref{T:transfer1} we can express $C_n$
asymptotically as
\begin{equation}\label{E:Catalan}
C_{n}\sim \frac{4^{n}}{\sqrt{\pi n}}
\left(\frac{1}{n+1}-\frac{1}{8n(n+1)}+\frac{1}{128n^{3}}+\frac{5}{1024n^{4}}
+{O}(n^{-5})\right) \ .
\end{equation}
The generating function of the Catalan numbers is given by
\begin{equation}\label{E:catalan}
\Psi(y)=\sum_{n\ge 0} C_n y^n =\frac{1-\sqrt{1-4y}}{2y}
\end{equation}
having a branch-point singularity at $\frac{1}{4}$.
Lemma~\ref{L:ana} allows us to express the analytic continuation of
$\sum_{n\ge 0}{\sf S}_2(n)z^n$ via $\Psi$:
\begin{eqnarray}\label{E:oben}
\Xi_2(z)& = &\frac{1}{z^2-z+1}\,
\Psi(\left(\frac{z}{z^2-z+1}\right)^{2}) \\
 & = & \frac{1}{z^2-z+1}\
\frac{1-\sqrt{1-4
\left(\frac{z}{z^2-z+1}\right)^{2}}}{2\left(\frac{z}{z^2-z+1}\right)^2}
=\frac{1-\sqrt{1-4
\left(\frac{z}{z^2-z+1}\right)^{2}}}{2\frac{z^2}{z^2-z+1}} \ .
\end{eqnarray}
The explicit form of $\Xi_2(z)$ allows us to conclude that
$\rho_2=\frac{3-\sqrt{5}}{2}$ is the unique dominant singularity. We
denote the map $z\mapsto (\frac{z}{z^2-z+1})^2$ by $\vartheta$ and
compute the first terms of the Taylor series at $z=\rho_2$,
i.e.~where $\vartheta(\rho_2)=\frac{1}{16}$:
\begin{equation}
\vartheta(z) = \frac{1}{4}+\frac{5+3\sqrt{5}}{8}(z-\rho_2) +
(z-\rho_2)^2\,T(z) \ ,
\end{equation}
where $T(z)=\sum_{i\ge 0}c_i(z-\rho_2)^i$, $c_i\in\mathbb{R}$.
Analyzing $\Xi_2(z)$ in an intersection of an $\epsilon$-disc
around $\rho_2$ with $\Delta_{\rho_2}$ produces
\begin{equation}\label{E:what}
\Xi_2(z)  =  \frac{1-\sqrt{
\left(\frac{5+3\sqrt{5}}{2}\right)(\rho_2-z) -
(z-\rho_2)^2\,T(z)}}{2\frac{z^2}{z^2-z+1}}
\end{equation}
from which we immediately conclude
\begin{equation}
\Xi_2(\rho_2 z)=O(\Psi(4z)) \ .
\end{equation}
Theorem \ref{T:transfer1} and the scaling property of Taylor coefficients
$[z^n]f(z)=\gamma^n [z^n]f(\frac{z}{\gamma})$ imply
\begin{equation}\label{E:coeff}
K [z^n]\, \Xi_2({\rho_2}z) \sim [z^n]\, \Psi({4}z), \quad
\text{\rm for some} \ K>0
\end{equation}
and we accordingly arrive substituting $\alpha=-\frac{1}{2}$ at
\begin{eqnarray*}
\quad [z^n]\, \Xi_2(z)  =
\frac{K}{\sqrt{n}}
\left(\frac{1}{n+1}-\frac{1}{8n(n+1)}+\frac{1}{128n^{3}}
+\frac{5}{1024n^{4}}
+{O}(n^{-5})\right)\, \left(\frac{3+\sqrt{5}}{2}\right)^{n},
\end{eqnarray*}
for some $K>0$. Via Theorem~\ref{T:cool1}, the coefficients ${\sf
S}_2(n)$ are explicitly known and we compute $K=1.9572$ from which the
theorem follows. $\square$.

%%%
%%%
%%%%%%%%%%%%%%%%%%%%%%%%%%%%%%%%%%%%%%%%%%%%%%%%%%%%%%%%%%%%%%%%%%%%%%%%%%
%%%
{\bf Acknowledgments.}
%%%
%%%%%%%%%%%%%%%%%%%%%%%%%%%%%%%%%%%%%%%%%%%%%%%%%%%%%%%%%%%%%%%%%%%%%%%%%%
%%%
We are grateful to Prof.~Jason Z.~Gao for helpful comments. Many thanks
to J.Z.M.~Gao
for helping to draw the figures. This work was supported by the 973 Project,
the PCSIRT Project of the Ministry of Education, the Ministry of Science and
Technology, and the National Science Foundation of China.

\bibliography{a}
\bibliographystyle{plain}

%%%
%%%%%%%%%%%%%%%%%%%%%%%%%%%%%%%%%%%%%%%%%%%%%%%%%%%%%%%%%%%%%%%%%%%%%%%%%%
%%%

\end{document}